\documentclass[a4paper,10pt]{article}

\usepackage[]{amsmath, amssymb, amsthm, array, 
 xcolor, pstricks, multirow}
\usepackage[]{tabularx}
\usepackage[]{float}
\usepackage[english]{babel}
\usepackage{rotating}
\usepackage{makeidx}
\usepackage{hyperref}
\usepackage{longtable}
\usepackage{url}
\usepackage[normalem]{ulem}

\usepackage{lineno}

\usepackage{fancybox} 

\newtheorem{theorem}{Theorem}[section]
\newtheorem*{theorem*}{Theorem}
\newtheorem*{remark*}{Remark}
\newtheorem{proposition}[theorem]{Proposition}
\newtheorem{lemma}[theorem]{Lemma}

\newtheorem{definition}[theorem]{Definition}

\theoremstyle{definition}
\newtheorem{remark}[theorem]{Remark}

\newtheorem{example}[theorem]{Example}

\newcommand{\N}{{\mathbb N}}     

\newenvironment{conditions}
{%
	\begin{list}{\rm (\theenumi)}%
	{\noindent%
		\usecounter{enumi}%
		\setlength{\topsep}{2pt}%
		\setlength{\partopsep}{0pt}%
		\setlength{\itemsep}{2pt}%
		\setlength{\parsep}{0pt}%
		\setlength{\leftmargin}{2.5em}%
		\setlength{\labelwidth}{1.5em}%
		\setlength{\labelsep}{0.5em}%
		\setlength{\listparindent}{0pt}%
		\setlength{\itemindent}{0pt}%
	}%
}%
{\end{list}}%

{\end{list}}%

\newenvironment{conditionsiii}
{%
	\begin{list}{\rm (\roman{enumi})}%
	{\noindent%
		\usecounter{enumi}%
		\setlength{\topsep}{2pt}%
		\setlength{\partopsep}{0pt}%
		\setlength{\itemsep}{2pt}%
		\setlength{\parsep}{0pt}%
		\setlength{\leftmargin}{2.5em}%
		\setlength{\labelwidth}{1.5em}%
		\setlength{\labelsep}{0.5em}%
		\setlength{\listparindent}{0pt}%
		\setlength{\itemindent}{0pt}%
	}%
}%
{\end{list}}%

{\end{list}}%

{\end{list}}%

{\end{list}}%

{\end{list}}%

\begin{document}

\title {On Lattices of Regular Sets of Natural  Integers Closed under Decrementation}
%

\author{
Patrick C\'egielski%
\footnote{Partially supported by TARMAC ANR agreement
12 BS02 007 01.}%
\newcounter{thanks}
\setcounter{thanks}{\value{footnote}}
\footnote{LACL, EA 4219, Universit\'e Paris-Est Cr\'eteil, France.}
\and
Serge Grigorieff%
\footnotemark[\value{thanks}]
\footnote{LIAFA, CNRS and Universit\'e Paris-Diderot, France.}
\newcounter{fnnumber}
\setcounter{fnnumber}{\value{footnote}}
\and
Ir\`ene  Guessarian%
\footnotemark[\value{thanks}]
\footnotemark[\value{fnnumber}]
\footnote{Emeritus at UPMC Univ Paris 6.}
\newcounter{fnnnumber}
\setcounter{fnnnumber}{\value{footnote}}
}
%


\maketitle

\bibliographystyle{plain}

\begin{abstract}
We consider lattices of regular sets of non negative integers,
i.e. of sets definable in Presbuger arithmetic.
We prove that if such a lattice is closed under decrement
then it is also closed under many other functions:
quotients by an integer, roots, etc.

\end{abstract}

\section{Introduction}
%
\subsection{Roadmap}
We follow the terminology according to which a function $f:\N\to\N$ is non decreasing
if $a\leq b\ \Rightarrow\ f(a)\leq f(b)$ for all $a,b\in\N$.

We prove in this paper the following result:
\begin{theorem}\label{thm:MAIN}
Let $f:\N\longrightarrow\N$ be a non decreasing function.
The following conditions are equivalent:
\begin{conditions}
\item
Every lattice $\+L$ of regular subsets of $\N$
which is closed under decrement
(i.e. $L\cap L'$, $L\cup L'$ and $L-1$ are in $\+L$
whenever $L,L'\in \+L$) is also closed under $f^{-1}$
(i.e. $L\in\+L$ implies $f^{-1}(L)\in\+L$).
\item
The function $f$ satisfies the following properties:\\
(i)\phantom{i} $f(a)\geq a$ for all $a\in\N$,\\
(ii) $f(a)-f(b)\equiv 0 \pmod{(a-b)}$
for all $a,b\in\N$.
\end{conditions}
Particular exemples of such functions $f$ are division by $n$ and $n$-root
for any $n\geq1$.
\end{theorem}

This problem, for finite sets and division by $n$,
was submitted to us by Jean-\'Eric Pin \& Zolt\'an \'Esik, \cite{EsikPin2011}.
Jean-\'Eric Pin \& Pedro Silva announce,
in the framework of profinite topologies and uniformly continuous fonctions,
a result related to our theorem \ref{thm:MAIN}  (see \cite{PinSilva2011,PinSilva}).

\smallskip
Any regular subset $L$ of $\N$ is ultimately periodic
(cf. Lemma~\ref{l:ultimately periodic}).
For an arithmetic progression $L$, the fact that $f^{-1}(L)$
is a union of decrements of $L$ is an easy result
(cf. Proposition~\ref{p:progression}).
Difficulties arise with:
\begin{conditions}
\item
the finite set coming from
the grouping of arithmetic progressions which constitutes the
periodic part of $L$,
\item
the other finite set before periodicity
(these two finite sets are the sets $B$ and $A$ of
Proposition~\ref{thm:Myhill}).
\end{conditions}

Prior to the general result
(cf. Theorems~\ref{thm:main2to1} \& \ref{thm:main2}),
we prove particular instances, namely division by $n$ and $n$th root, which give a clearer insight to the proof
(cf. Theorems \ref{thm:quotient}
and \ref{thm:root}).

%
\subsection{Lattices closed under decrementation}
%
We recall some definitions and fix some notation.
\begin{definition}
A lattice $\+L$ over a set $X$ is any non empty family
of subsets of $X$ such that $L\cup M$ and $L\cap M$ are in $\+L$
whenever $L,M$ are in $\+L$.
\end{definition}

\begin{definition}
Let $L$ be a subset of $\N$, $i\in\N$ and $k\in\N\setminus\{0\}$.
The sets
$$
L-i = \{x\in\N \mid x+i\in L\}
\ , \
L\div k = \{x\in\N \mid kx \in L\}
\ , \
\sqrt[k]{L} = \{x\in\N \mid x^k\in L\}
$$
are respectively called the $i$-\emph{decrement},
$k$-\emph{quotient} and $k$-\emph{root} of $L$.
Observe that the $i$-decrement is defined as a subset of $\N$,
excluding negative integers.

Let ${\+D}(L)$ denote the family $\{L-i \mid i\in\N\}$
of decrements of $L$.
\end{definition}

\begin{example}\label{ex4N}
1) Let $L=\{5,6\}+4\N=\{5,6,9,10,13,14,\ldots\}$, then $L\div 2 = 3+2\N=\{3,5\}+4\N$.
Moreover, for any integer $x$, $x^2\equiv 0\pmod 4$ or $x^2\equiv 1\pmod 4$, hence
\begin{eqnarray*}
x^2\in \{5,6\}+4\N
&\iff& x^2\geq5\wedge x^2\equiv 1\pmod 4\\
&\iff& x\geq3\wedge x\equiv 1,3\pmod 4\\
&\iff& x\in \{3,5\}+4\N
\end{eqnarray*}
Hence also $\sqrt{L}=\{3,5\}+4\N=L\div 2\;$.

2)  Let $L=\{1,2\}+4\N$, then $L\div 3 = \{2,3\}+4\N$ and $\sqrt{L} =\{1,3\}+4\N$.
\end{example}

The following results are straightforward.
\begin{proposition}[Composing decrements]
\label{p:successive decrement}
$(L-i)-j=L-(i+j)$.
\end{proposition}
\begin{proposition}\label{p:sublattice}
For $L\subseteq\N$ let $\+L(L)$ be the family of
sets of the form $\bigcup_{j\in J}\bigcap_{i\in I_j}(L-i)$
where $J$ and the $I_j$'s are finite non empty subsets of $\N$.
Then the family $\+L(L)$ is the smallest sublattice of $\+P(\N)$
containing $L$ and closed under decrement.
\end{proposition}
\begin{proof}
Observe that
$\left(\bigcup_{j\in J}\bigcap_{i\in I_j}(L-i)\right)-k
=\bigcup_{j\in J}\bigcap_{i\in I_j}\big(L-(i+k)\big)$.
\end{proof}

\subsection{Regular sets of natural integers}
%
\begin{definition}\label{d:periodic}
1. A set $L\subseteq\N$ is periodic with  period $r$ if,
for every $x$, $x\in L \Longrightarrow x+r\in L$.
\\
2. A set $L\subseteq\N$ is ultimately periodic with  period $r$ if
there exists $q\in\N$  such that
$L\cap\{x\mid x\geq q\}$  is periodic with  period $r$,
i.e. for every $x\geq q$, $x\in L \Longrightarrow x+r\in L$.
\end{definition}

As we here we work with a semigroup and not a group, namely $(\N,+)$,
the definition of periodicity is not given by an equivalence
$x\in L \Longleftrightarrow x+r\in L$ but by an implication
$x\in L \Longrightarrow x+r\in L$.

\smallskip
Regular subsets of $\N$ are subsets which are recognized by finite automata
in unary notation (cf. \cite{eilenberg}, pages 100--103).
Here, we will only use the following classical characterization of regular subsets
of $\N$ which goes back to Myhill, 1957 \cite{myhill}.
Recall that an arithmetic progression is a subset of $\N$ of the form $a+r\N$.

\begin{proposition}\label{thm:Myhill}
Let $L\subseteq N$.
The following conditions are  equivalent:
\begin{conditionsiii}
\item
$L$ is regular,
\item
$L$ is the union of a finite set with finitely many arithmetic progressions,
\item
$L=A \cup (q+B+r\N)$,
where $q\in\N$,  $r\in\N\setminus\{0\}$, $A\subseteq\{0,1,\ldots, \max(0,q-1)\}$
and $B\subseteq\{0,1,\ldots, r-1\}$.
\end{conditionsiii}
\end{proposition}
Observe that in case $B=\emptyset$, the set $A \cup (q+B+r\N)$
reduces to the finite set $A$.
The following lemmas will be useful.

\begin{lemma}\label{l:ultimately periodic}
Any regular set is ultimately periodic
and its family of decrements is finite.

More precisely, suppose $L=A \cup (q+B+r\N)\subseteq\N$
where $q\in\N$, $r\in\N\setminus\{0\}$,
$A\subseteq\{0,1,\ldots, q-1\}\cap\N$,
and $B\subseteq\{0,1,\ldots, r-1\}$.
Then
\begin{conditions}
\item
$\forall x\geq q \  (x\in L  \Longleftrightarrow x+r\in L)$
\item
The family $\+D(L)$ of  decrements of $L$ is equal to
$\{L-i \mid 0\leq i < q+r\}$.
\end{conditions}
\end{lemma}
\begin{proof}
(1) \  Let $x\geq q$,
so that $x=q+i+kr$ for some $0\leq i < r$, $k\geq 0$.
Then $x\in L=A\cup(q+B+r\N)\iff q+i+kr\in q+B+r\N\iff i\in B$.
Similarly, $x+r\in L\iff q+i+(k+1)r\in q+B+r\N\iff i\in B$.
Thus, $x\in L\iff x+r\in L$.
\\
(2) \ Let $j\geq q$. Then $j=q+i+kr$ for some $0\leq i < r$, $k\geq 0$.
For any $x\in\N$, we have
$x\in L-j \Leftrightarrow x+j\in L
\Leftrightarrow x+q+i+kr\in L
\Leftrightarrow x+q+i\in L
\Leftrightarrow x\in L-(q+i)$,
the third in place of equivalence being obtained by applying $k$ times point $(1)$.
\end{proof}

\begin{example}(Example \ref{ex4N} continued)\label{ex:decrement}
1) For $L=\{5,6\}+4\N$, the set $\+D(L)$ consists of 7 sets $L, L-1=\{4,5\}+4\N,\ldots, L-5=\{0,1\}+4\N, L-6=\{0,3\}+4\N, L-7=\{2,3\}+4\N=L-3$.

2) If $L=\{1,2\}+4\N$, then $\+D(L)=\{L, \{0,1\}+4\N, \{0,3\}+4\N, \{2,3\}+4\N\}$.
\end{example}

In case of an arithmetic progression,
Proposition~\ref{p:sublattice} can be simplified.
\begin{lemma}\label{l:progression} 
Let $L=q+r\N$ be the range of an arithmetic sequence, $r>0$.
\\
1. The family $\+D(L)$ of decrements of $L$ is equal to
$$
\+D(L)=\{s+r\N \mid 0\leq s\leq\max(r-1,q)\}
=\{L-j \mid 0\leq j\leq\max(r-1,q)\}\ .
$$
2. The smallest lattice $\+L(L)$ containing $L$
and closed under decrement is equal to the family of sets
\begin{conditions}
\item[(i)]
$ \{A+r\N \mid A\subseteq \{0,\ldots,\max(r-1,q)\}$
if $r\geq 2$,
\item[(ii)]
$\{s+\N \mid 0\leq s\leq q\}$ if $r=1$.
\end{conditions}
In particular,
every nonempty set of $\+L(L)$ is a finite union
of decrements of $L$,
and the empty set is in $\+L(L)$   just in case $r\geq 2$
(obtained with $A=\emptyset$).
\end{lemma}

\begin{proof}
1. In case $j\leq q$ then $L-j=s+r\N$ with $s=q-j\leq q$.
If $j\geq q$, i.e. $j=q+i+kr$ with $0\leq i<r$ and $k\in\N$,
then $L-j=\{x\in\N\mid x+(q+i+kr)\in q+r\N\}
=\{x\in\N\mid x+i\in r\N\}=r\N-i$.
If $i=0$ then $L-j=r\N=L-q$.
If $0<i<r$ then $L-j=r\N-i=(r-i)+r\N=L-(q-(r-i))$.
\\
2. Observe that the intersection of two sets in the family $\+D(L)$
is either empty (possible in case $r\geq2$ only)
or equal to the smallest one.
Then apply Proposition \ref{p:sublattice}, noting that for $r=1$, $A+r\N=\min (A)+\N$.
\end{proof}
%
%
\section{Closure under quotient and root}\label{sec:quotient}

%
The following result was suggested for lattices of finite sets
by \'Esik \& Pin  \cite{EsikPin2011}.
\begin{theorem}\label{thm:quotient}
Any lattice of regular subsets of $\N$ which is closed under decrement
is also closed under $k$-quotient, for $k\in \N\setminus\{0\}$.
\end{theorem}

\begin{proof} 
The case $k=1$ is trivial.
We prove the theorem by induction on $k\geq1$.
For pedagogical reasons, we explicit the case $k=2$.
\smallskip\\
{\it Case $k=2$.}
Consider some $L\in\+L$ and let
$J_a=(L-a)\cap\bigcap_{i \in L-a} (L-i)$
for any $a\in\N$.
By Lemma~\ref{l:ultimately periodic}),
there are finitely many distinct sets $(L-i)$'s,
so that $J_a$ is a finite intersection of decrements of $L$.
The assumed closure properties of $\+L$ insure that $J_a\in\+L$.

In case $a\in L\div 2$, i.e. $2a\in L$,
the following properties are true.
\begin{conditions}
\item
$a\in J_a$.
In fact, $a\in L-i$ for any $i\in L-a$
and $a\in L-a$ because $2a\in L$.
\item
$J_a\subseteq L\div 2$.
Indeed, if $b\in J_a$ then $b\in L-a$
and $b$ is in all the ($L-i$)'s, for $i\in L-a$.
Letting $i=b$, we get
$b\in L-b$, i.e $2b\in L$ and $b\in L\div 2$.
\end{conditions}
Since there are finitely many $L-a$'s,
there are finitely many $J_a$'s.
Using closure under finite union, we see that
$K=\bigcup_{a\in L\div 2} J_a$ is in $\+L$.
Clearly, $K=L\div 2$ because each element  $a\in L\div 2$ is in $J_a$ and each $J_a$ is included in $L\div 2$.
\smallskip
\\
{\it Inductive case.}
Assuming $\+L$ is closed under $k$-quotient,
we prove that it is closed under $(k+1)$-quotient.
For $L\in\+L$, set
$J_a=\left((L-a)\div k\right)\cap\bigcap_{i \in (L-a)\div k} (L-ki)$
for any $a\in\N$.
By Lemma~\ref{l:ultimately periodic},
there are finitely many distinct $(L-i)$'s,
so that $J_a$ is a finite intersection of decrements of $L$ and of
a $k$-quotient of $L$.
The assumed closure properties of $\+L$ and induction hypothesis
insure that $J_a\in\+L$.

In case $a\in L\div (k+1)$, i.e. $(k+1)a\in L$,
the following properties are true.
\begin{conditions}
\item
$a\in J_a$.
In fact, $a\in L-ki$ for any $i\in (L-a)\div k$.
Also, since $(k+1)a\in L$, we have $ka\in L-a$
hence $a\in (L-a)\div k$.
\item
$J_a\subseteq L\div (k+1)$.
If $b\in J_a$ then $b\in(L-a)\div k$
and $b$ is in all the $L-ki$'s, for $i\in(L-a)\div k$.
Letting $i=b$, we get $b\in L-kb$, i.e $(k+1)b\in L$ and $b\in L\div (k+1)$.
\end{conditions}
Since there are finitely many ($L-a$)'s,
there are finitely many $(L-a)\div k$'s hence finitely many $J_a$'s.
Using closure under finite union, we see that the set
$K=\bigcup_{a\in L\div (k+1)} J_a$ is in $\+L$.
Clearly, $K=L\div (k+1)$ because
each element  $a\in L\div (k+1)$ is in $J_a$
and each $J_a$ is included in $L\div (k+1)$.
\end{proof} 
\begin{theorem}\label{thm:root}
Any lattice of regular subsets of $\N$ which is closed under decrement
is also closed under $k$-root, for $k\in \N\setminus\{0\}$.
\end{theorem}
\begin{proof}
Adapt the above proof:
substitute $\times$ and division for + and subtraction, so that
$L-i$ becomes $L\div i$.
In the argument,
finiteness of the family $\{L-i\mid i\in\N\}$ is replaced by that of
$\{L\div k\mid k\in\N\setminus\{0\}\}$ which holds since,
by Lemma \ref{l:ultimately periodic} and Proposition \ref{p:sublattice},
$\+L(L)$ is always finite when $L$ is regular. 
\end{proof}

\begin{example}(Examples \ref{ex4N} and \ref{ex:decrement} continued)\label{ex:sqrt4N}
 If $L=\{1,2\}+4\N$, then $L\div 3=\{ 2+4\N\}\cup\{ 3+4\N\}=L-3$ and
$\sqrt{L} =\{1,3\}+4\N=(L-5)\cup(L-3)$.
 
For $L=\{5,6\}+4\N$, we have $L\div 2=\sqrt{L} =\{3,5\}+4\N=\big((L-2)\cap(L-3)\big)\cup\big(L\cap (L-1)\;\big)$.
\end{example}

\section{More induced closures}\label{sec:main}
%
We extend closure under quotient (cf. Theorem~\ref{thm:quotient})
and under $n$-root (cf. Theorem~\ref{thm:root})
to a more general class of functions $f:\N\to\N$.
Given a regular set $L\subseteq\N$ and $n\in\N$, 
the set $L - n = \{x\in\N \mid x + n \in L\}$ is regular.
Also, by Lemma~\ref{l:ultimately periodic},  the family $\{L - n | n \in\N\}$ is finite.

\begin{lemma}\label{l:main2to1}
For any set $L \subseteq \N$ and for any function $f :\N\to\N$ such
that $f(x)-f(y) \in (x-y)\N$ for every $x, y \in \N$,
and such that $f(x) \geq x$ for every $x \in\N$,
we have:
\begin{equation}\label{eq:f-1L}
f^{-1}(L) \ =\ \bigcup_{a\in f^{-1}(L)}\left(\bigcap_{n\in L-a} L-n\right)
\end{equation}
\end{lemma}
\begin{proof}
Let us first consider $a \in f^{-1}(L)$.
Notice that for every $n \in L-a$, we have $a+n \in L$ and thus $a \in L-n$.
We deduce that $a$ is in $\bigcap_{n\in L-a} L-n$ and the
inclusion $\subseteq$ is proved.

\smallskip
For the other inclusion, let $a \in f^{-1}(L)$ and $b \in\bigcap_{n\in L-a} L-n$.
By the assumption on $f$, there exists $k \in\N$ such that
$f(a)-f(b) = k(a-b)$. 
Assume by contradiction that $f(b)\notin L$.
Since $f(a) \in L$ we get $f(a) \neq f(b)$, and in particular $a \neq b$.

Assume first that $a < b$. We consider the minimal natural number $r \in\N$
such that $f(a) + r(b-a) \notin  L$.
Note that such a natural number exists since
$f(a) + k(b-a) = f(b) \notin  L$.
Moreover, since $f(a) \in L$ we get $r\geq 1$.
By minimality of $r$, we get $f(a) + (r-1)(b-a) \in L$.
Thus, $n + a \in L$ with $n = f(a) + r(b-a)-b$.
Since $f(a)\geq a$, we get $n\geq (r-1)(b-a)\geq 0$. 
Now $n + a \in L$ implies $n \in L-a$ and 
thus $b \in L-n$; hence $n+b=f(a) + r(b-a) \in L$,
contradicting the definition of $r$.

Assume next that $a > b$  
and consider the minimal natural number $r \in\N$
such that $f(b) + r(a-b) \in L$.
Again, such a natural number exists since $f(b) + k(a-b) = f(a) \in L$.
Moreover, since $f(b) \notin  L$, we get $r\geq 1$.
Let $n = f(b)-b + (r-1)(a-b)$.
Since $f(b)\geq b$ and $a-b\geq 0$ we get $n\geq 0$.
Moreover, as $n + a = f(b) + r(a-b) \in L$, we get $n \in L-a$.
Thus, $b \in L-n$ and we get $n + b \in L$. 
That means $n+b=f(b) + (r-1)(a-b) \in L$ which contradicts the minimality of $r$.

We have proved by contradiction that $f(b) \in L$.
Thus, $b \in f^{-1}(L)$ and we get the other inclusion.
\end{proof}
We  can now prove the $(2)\Rightarrow(1)$ implication 
of our main theorem~\ref{thm:MAIN}.
\begin{theorem}\label{thm:main2to1}
Let $f:\N\to\N$ be non decreasing and such that
(i) $f(a)\geq a$
and (ii) $f(a)-f(b)\equiv 0 \pmod{(a-b)}$
for all $a,b\in\N$.
Every lattice of regular subsets of $\N$ closed under decrement
is also closed under $f^{-1}$.
\end{theorem}
\begin{proof}
Let $\+L$ be a lattice of regular sets closed under decrement
and let $L\in\+L$. Consider the representation of $f^{-1}(L)$ given by
formula~\eqref{eq:f-1L} of Lemma~\ref{l:main2to1}.
In order to ensure that $f^{-1}(L)$ belongs to the lattice $\+L$, we have to show that
both the intersection and the union are finite:  
since $L$ is regular, the family $\{L - n | n \in\N\}$ is finite by Lemma~\ref{l:ultimately periodic}; 
this concludes the proof.
\end{proof}
\begin{remark}
Every non decreasing polynomial with integral coefficients
mapping $\N$ into $\N$
satisfies the conditions of Theorem~\ref{thm:main2to1}.
Thus, Theorems~\ref{thm:quotient} and \ref{thm:root}
are  consequences of Theorem~\ref{thm:main2to1}; 
their proof gives
 a first idea and a better understanding to prove the  more general Theorem~\ref{thm:main2to1}).
\end{remark}
%
%
\section{About arithmetic progressions}\label{sec:arith}

For arithmetic progressions we sharpen Theorem~\ref{thm:main2to1}
and give a simpler proof.
\begin{proposition}\label{p:progression}
Let $f:\N\to\N$ non decreasing  be such that for all $a,b\in\N$
(i) $f(a)\geq a$
and (ii) $f(a)-f(b)\equiv 0 \pmod{(a-b)}$.
For every arithmetic progression $L=q+r\N$, with $q,r\in\N$, $r\geq1$,
the following conditions hold:
\begin{conditions}
\item
$f^{-1}(L)$ is the union of at most $r$ decrements of $L$,
\item
the smallest lattice $\+L(L)$ closed under decrement and such that 
$L\in\+L$ is closed under $f^{-1}$.
\end{conditions}
\end{proposition}
\begin{proof}
(1) If $f(a)\in q+r\N$ then, using monotonicity of $f$ and property (ii),
for every $k\in\N$ there exists $\ell\in\N$
such that $f(a+kr)=f(a)+\ell r$ hence $f(a+kr)\in q+r\N$ and $a+r\N\subseteq f^{-1}(L)$.
Thus, $f^{-1}(L)=\bigcup_{a\in f^{-1}(L)}a+r\N$.
Now, if $a<b$ and $a\equiv b\ (\bmod\ r)$ then
$b+r\N\subseteq a+r\N$.
Hence the last equality can be rewritten
$f^{-1}(L)=\bigcup_{a\in M}a+r\N$
where $M$ picks the minimum element of
$f^{-1}(L)\cap (i+r\N)$ for each $i$ such that  $0\leq i<r$
and $f^{-1}(L)\cap (i+r\N)$ is nonempty.
In particular, $M$ has at most $r$ elements.

It remains to show that, for each $a\in M$,
the set $a+r\N$ is a decrement of $L$.
Using Lemma~\ref{l:progression}, this amounts to show that
$a\leq\max(r-1,q)$ for each $a\in M$.
Let $a\in M$, $a=\min\left(f^{-1}(L)\cap (i+r\N)\right)$
with $0\leq i<r$.
By way of contradiction, supposing $a>\max(r-1,q)$,
so that $a-r\in\N$,
we show that $f(a-r)\in L$.
\\
{\it Case $q< r$.} Then $\max(r-1,q)=r-1<a$.
Since $f(a)\in L$ we have $f(a)=q+kr$ for some $k\in\N$.
Using  property (ii), we get
$f(a) \equiv f(a-r) \equiv q \ (\bmod\ r)$.
Since $q<r$, this yields $f(a-r)=q+\ell r$ for some $\ell\in\N$
and thus $f(a-r)\in L$.
\\
{\it Case $q\geq r$.} Then $\max(r-1,q)=q$ and $a>q\geq r$.
Let $q=i+kr$ with $0\leq i<r$ and $k\geq1$.
As above, $f(a-r)\equiv f(a)\equiv i\ (\bmod\ r)$
hence $f(a-r)=i+\ell r$ for some $\ell\in\N$.
Now, $f(a-r)\geq a-r$ by (i) hence $i+\ell r\geq a-r>q-r=i+(k-1)r$
so that $\ell\geq k$.
Thus, $f(a-r)=i+\ell r=q-kr+\ell r= q+(\ell-k)r\in q+r\N=L$.
\\
In both cases, we have $f(a-r)\in L$,
contradicting the minimality of $a$ in the intersection of
its congruence class modulo $r$ with $f^{-1}(L)$.

\medskip
(2) By Lemma \ref{l:progression}
any set $K$ in $\+L(L)$ is of the form $K=A+r\N$ with
$A\subseteq \{0,\ldots,\max(r-1,q)\}$, 
hence $\forall a\in A\quad a\leq \max(r-1,q)$.
Then  $f^{-1}(K)=\cup_{a\in A}f^{-1}(a+r\N)$.
By (1), each $f^{-1}(a+r\N)$ is of the form $A_a+r\N$
with $A_a\subseteq\{0,\ldots,\max(r-1,a)\}$.
Hence $f^{-1}(K)=\cup_{a\in A}\big(A_a+r\N\big)
=(\cup_{a\in A}A_a)+r\N$, 
and $\cup_{a\in A}A_a$ is a subset of $\{0,\ldots,\max(r-1,q)\}$.
When $r\geq 2$, this concludes the proof  that $f^{-1}(K)\in  \+L(L)$
by Lemma \ref{l:progression} 2(i). 
If $r=1$, we must check also that $f^{-1}(K) \not= \emptyset$: 
indeed $K=a+\N$ by Lemma \ref{l:progression} 2(ii),
as $f(a)\geq a$ by hypothesis (i), $f(a)\in a +\N$ and
$a\in f^{-1}(K) $ which  is non empty;
this concludes the proof that $f^{-1}(K)\in  \+L(L)$ for the case $r=1$.
%
\end{proof}

\begin{remark}
The statement of Proposition~\ref{p:progression} is sharper than
that of Theorem~\ref{thm:main2to1} applied to arithmetic progressions.
In fact, the proof of Proposition~\ref{p:progression} also shows that,
for an arithmetic progression $L$, 
the lattice $\+L(L)$ is
the smallest join-semilattice containing $L$ and closed under decrement.
\end{remark}
\begin{remark}
The proof of Proposition~\ref{p:progression} cannot be extended to regular sets, not even to periodic sets. Let  $f\colon n\mapsto n^2$ and $L$ the periodic set
$\{0 ,4,8\}+3\N=\{0,3,4\}\cup(6+\N)$.
Then $ f^{-1}(L)=\sqrt{L} =\N\setminus\{1\}=L-4$ is a decrement of $L$.
However, this result cannot be obtained by the proof of Proposition~\ref{p:progression}
 because this proof relies on the fact that, whenever $a\in f^{-1}(L)$, then 
 $a+r\N$ is a decrement of $L$; here however, $2+3\N$ is not a decrement of $L$ and does not even belong to $\+L(L)$).
Indeed, $\+D(L)$ consists here of $L, L-1, L-2, L-3,L-4, L-5, L-6=\N$,
all of which are of the form $D_i\cup(6+\N)$ with $D_i$ a finite set. Thus, $2+3\N$ cannot be obtained by finite unions, intersections and decrements of such sets, all of which contain {\it all} the integers larger than 6.
\end{remark}
\begin{remark}
 Proposition~\ref{p:progression}
does not hold  for finite sets, nor  general regular sets, nor periodic sets: unions of decrements are not sufficient to obtain $ f^{-1}(L)$,  intersections are needed.
\\
Consider $f\colon n\mapsto n^2$.
Let $L=\{5,6\}+4\N$ (periodic); then  (cf. Example  \ref{ex:sqrt4N}) $L\div 2= \sqrt{L}=\{3,5\}+4\N=\big((L-2)\cap(L-3)\big)\cup\big(L\cap (L-1)\big)$
cannot be obtained as a union of decrements of $L$: in order to obtain 5, we must include either   $L,\ L-1,\ L-4$ or $L-5$, but each of these decrements contains numbers not
 in $\sqrt{L} $ (respectively 6, 4, 2 and 0) which must be excluded by a suitable intersection.
 
Let $L=\{1,2\}$ then $f^{-1}(L)=\{1\}$ ;
 the  decrements of $L$ are the sets
$\{1,2\},\{0,1\},\{0\},\emptyset$,
no union of which is $f^{-1}(L)$,
intersections are required to get $f^{-1}(L)$.


This is why the proof in both the general and the finite case does
exclude the elements which are not in $f^{-1}(L)$ by using carefully chosen intersections.
\end{remark}

%
\section{Characterizing induced closures}\label{sec:main2}
%
We characterize the functions $f:\N\to\N$ such that
closure under decrement yields closure under $f^{-1}$.

\begin{theorem}\label{thm:main2}
Let $f:\N\to\N$.
The following conditions are equivalent.
\begin{conditionsiii}
\item
Every lattice of regular subsets of $\N$ closed under decrement
is closed under $f^{-1}$.
\item
For every finite subset $L$  of $\N$, the lattice $\+L(L)$ is closed under $f^{-1}$.
\item  For every arithmetic progression $L=q+r\N$, $r>0$, 
the lattice $\+L(L)$ is closed under $f^{-1}$.
\item
The map $f$ is non decreasing and satisfies $f(a)\geq a$ and
$f(a)-f(b)\equiv 0 \pmod {(a-b)}$ 
for all $a,b\in \N$.
\end{conditionsiii}
\end{theorem}
\begin{proof}
{ (iv) $ \Rightarrow$  (i)}.  This is Theorem~\ref{thm:main2to1}.
\\
{ (i)  $\Rightarrow$  (ii)}.  Finite sets are regular sets.
\\
{ (i)  $\Rightarrow$  (iii)}.  Arithmetic progressions
are regular sets.
\\
(ii) $\Rightarrow$  (iv). 
We first prove that $f(a)\geq a$, for all $a\in \N$. 
Let $a\in\N$ and $L=\{f(a)\}$.
Observe that the smallest lattice containing the set $\{f(a)\}$
and closed under decrementation is the family of subsets
of $\{0,1,\ldots, f(a)-1, f(a)\}$.
As a consequence,
all elements of $f^{-1}(L)$ must be less than $f(a)$.
In particular $a\leq f(a)$, since $a\in f^{-1}(L)$.

We prove now that $f(a)-f(b)\in(a-b)\N$ for all $a,b\in \N$
such that $a>b$.
In particular, $f$ is monotone non decreasing and
$f(a)-f(b)\equiv0\  \pmod {(a-b)}$.
We argue by contradiction. Suppose that $f(a)\notin f(b)+(a-b)\N$.
Let
$$
\ell=\left\lfloor\frac{f(a)-a}{a-b}\right\rfloor
,\quad
k=\left\lfloor\frac{f(a)}{a-b}\right\rfloor
,\quad
L\ =\ \{f(a) -j(a-b) \mid 0\leq j\leq k\}\ .
$$
Since $f(a)\geq a$, we have $\ell\geq0$;  moreover, 
$$
k=\left\lfloor\frac{f(a)}{a-b}\right\rfloor
=\left\lfloor\frac{f(a)-a}{a-b}+1+\frac{b}{a-b}\right\rfloor
$$
hence $k\geq\ell+1$.

For $j\in\{0,\ldots,k\}$,
$f(a)\neq f(b)+j(a-b)$ hence $f(b)\neq f(a)-j(a-b)$.
Thus, $f(b)\notin L$ and $b\notin f^{-1}(L)$.
Of course, $f(a)\in L$ and $a\in f^{-1}(L)$.
To get a contradiction, we show that $f^{-1}(L)$ is not in $\+L(L)$.
Since $f^{-1}(L)$ contains $a$ but not $b$,
it suffices to show that every set $X\in\+L(L)$ which contains $a$
also contains $b$.
Since $\+L(L)$ is generated by the $L-i$'s,
we reduce to show that,
for all $i$, if $a$ is in $L-i$ then so is $b$.
Now, using the definition of $\ell$, for all $i\in \N$
$$\begin{array}{rcl}
\ a\in L-i
&\iff& \exists \alpha\in\{0,\ldots,k\}\ \ a=f(a)-\alpha(a-b)-i
\\
&\iff& \exists \alpha\in\{0,\ldots,k\}\ \ i=f(a)-a-\alpha(a-b)
\\
&\iff& \exists \alpha\in\{0,\ldots,\ell\}\ \ i=f(a)-a-\alpha(a-b)
\\
\multicolumn{3}{l}
{\text{and, for $i$ associated to such an $\alpha\in\{0,\ldots,\ell\}$,}}
\\
L-i &=& \N\ \cap\ \{a+(\alpha-j)(a-b) \mid j\in\{0,\ldots,k\}\}
\\
\multicolumn{3}{l}
{\text{letting $j=\alpha$ and $j=\alpha+1$ (which is $\leq k$ since
 $\alpha\leq\ell<k$), we see that}}
\\
L-i &\supseteq& \{a,b\}\ .
\end{array}$$
This gives the required contradiction.

{(iii)  $\Rightarrow$ (iv).} 
Note first that if (iii) holds, $f$ cannot be constant:
indeed, for any constant function $f(x)=a$, there exists an arithmetic progression, namely $L=a+1+\N$,
such that the lattice $\+L(L)$ is not closed under $f^{-1}$.
In fact,  $f^{-1}(L)=\emptyset \notin \+L(L)$
because  all sets of $\+D(L)$ are of the form
$\{\ell+\N\ |\  0\leq \ell \leq a+1\}$,
hence all their finite intersections contain $a+1+\N$
and so are not empty.

By Lemma \ref{l:progression}, 
if $L=q+r\N$, with $q,r\in\N$, $r\geq1$,
then $\+L(L)$ is the family of sets of the form $B+r\N$
with $B\subseteq\{0,\ldots,\max(q,r-1)\}$, and $B\not=\emptyset$ if $r=1$.

First, we check that $f$ is non decreasing.
Let $a<b$ and let $L=f(a)+\N$.  Note that, since $r=1$, a  set $B+\N$ is equal to $\min(B)+\N$
and $\emptyset\not\in \+L(L)$,
hence $f^{-1}(L)=s+\N$, with $s\leq f(a)$; as $a\in f^{-1}(L)$, then $a\in s+\N$, i.e. $a\geq s$, and $f(s+\N)\subseteq L$.
Since $b>a\geq s$ we get $b\in s+\N$ and $f(b)\in L$ hence $f(b)\geq f(a)$.

Second, we show that $f(b)-f(a)\in(b-a)\N$ whenever $a<b$.
Let $L=f(a)+(b-a)\N$.
Then, in view of Lemma \ref{l:progression}, we may write  $ f^{-1}(L)=A+(b-a)\N$.
Since $a\in f^{-1}(L)$, we have $a+(b-a)\N\subseteq f^{-1}(L)$
hence $b=a+(b-a)\in f^{-1}(L)$, i.e. $f(b)\in L$, whence
for some $k$,  $f(b)=f(a)+k(b-a)$.

Finally, we show that $f(a)\geq a$ for all $a\in\N$.
Suppose that, for some $a$, $f(a)<a$.
Since  $a$ divides $f(a)-f(0)\leq f(a)<a$, we have $f(a) =f(0)$
hence $f$ is constant on $\{0,\ldots,a\}$ with value $<a$.
\\
{\it Case 1. There are infinitely many $a$'s such that $f(a)<a$.}
Then $f$ is constant, contradicting what was proved above.
\\
{\it Case 2. There is a largest $a$ such that $f(a)<a$.}
Then $f(x)=f(0)<a$ for $x\leq a$ and $f(x)\geq x>a$ for $x>a$.
Thus, $\emptyset\neq f^{-1}(a+\N)\subseteq(a+1)+\N$
is not in $\+L(a+\N)$,
 contradicting (iii).
\end{proof}

\section*{Acknowledgments}
  We thank the anonymous referees for their insightful reading and comments which helped
in   improving the paper.
%



\begin{thebibliography}{AA}
\setlength{\parsep}{0pt}
\setlength{\partopsep}{0pt}
\setlength{\parskip}{0pt}
\setlength{\itemsep}{0pt}
\bibitem{eilenberg} S. Eilenberg, Automata, languages and machines, vol. A,
Academic Press, New York, 1974.

\bibitem{EsikPin2011}
Z.~\'Esik, and J-\'E.~Pin.
\newblock {\em Personal communication.}
\newblock{July, 2011.}

 \bibitem{myhill} J. Myhill, Finite automata and the representation of events, Wright Air Development Command Tech. Rep. 5764, 1957, pp. 112-137.
 
 \bibitem{PinSilva2011}
J.-\'E. Pin and P.V. Silva, On profinite uniform structures defined by varieties of finite monoids,
International Journal of Algebra and Computation, 21, 2011, pp 295-314.

\bibitem{PinSilva}
J.-\'E. Pin and P.V. Silva, On  uniformly continuous functions for some profinite topologies, in preparation.

\bibitem{stansifer} R. Stansifer, 
PresburgerÕs Article on Integer Arithmetic: Remarks and Translation, Cornell tech reports,
 {\url{http://ecommons.library.cornell.edu/bitstream/1813/6478/1/84-639.pdf}}.
\end{thebibliography}
\end{document}